\documentclass[11pt]{article}

\usepackage{amssymb,amsthm,amsmath}
\usepackage{mathrsfs}
\usepackage{graphicx}
\usepackage[FIGTOPCAP,nooneline]{subfigure}
\usepackage[top=40mm,bottom=40mm]{geometry}

\usepackage{color}

\makeatletter
    
    \@addtoreset{equation}{section}
\makeatother

\theoremstyle{definition}
\newtheorem{theorem}{Theorem}[section]

\newtheorem{proposition}[theorem]{Proposition}

\usepackage{authblk}

\title{Self-similar motions and related relative equilibria in the $N$-point vortex system}

\author{Takeshi Gotoda \footnote{Graduate School of Mathematics, Nagoya University, Furocho, Chikusaku, Nagoya, Aichi, JAPAN  E-mail: gotoda@math.nagoya-u.ac.jp} }

\date{}

\begin{document}


\maketitle

\begin{abstract}
We study self-similar solutions of the point-vortex system. The explicit formula for self-similar solutions has been obtained for the three point-vortex problem and for a specific example of the four and five point-vortex problems. We see that the families consisting of these self-similar collapsing solutions are described by one-parameter families, and their collapse time and Hamiltonian are also expressed by functions of the same parameter. Then, the configurations at limit points of the parameter are in relative equilibria. For the many-vortex problem, we investigate the point-vortex system with the help of numerical computations. In particular, considering the case that $N - 1$ point vortices have a uniform vortex strength, we show that families of self-similar collapsing solutions continuously depend on the Hamiltonian and the self-similar solutions asymptotically approach relative equilibria as the Hamiltonian gets close to certain values. In addition, we prove the existence of relative equilibria for the four point-vortex system. We also investigate an example of seven point vortices with non-uniform vortex strengths and give numerical results for it.

\end{abstract}

\section{Introduction}

The motion of point vortices on a plane has been investigated by many researchers for a long time. The dynamics of point vortices is formulated as a Hamiltonian system \cite{Kirchhoff}: let $(x_m(t), y_m(t))$ for $m = 1,\cdots,N$ be the positions of point vortices and $\Gamma_m$ be their strengths. Then, the motion of $(x_m(t), y_m(t))$ is described by
\begin{equation*}
\Gamma_m\dfrac{\mbox{d} x_m}{\mbox{d} t} = \dfrac{\partial \mathscr{H}}{\partial y_m}, \qquad \Gamma_m\dfrac{\mbox{d} y_m}{\mbox{d} t} = -\dfrac{\partial \mathscr{H}}{\partial x_m} \label{pv-h}
\end{equation*}
with the Hamiltonian,
\begin{equation}
\mathscr{H} = - \frac{1}{2 \pi} \sum_{m=1}^N \sum_{n=m+1}^N \Gamma_m \Gamma_n \log{l_{mn}}, \label{Hamiltonian}
\end{equation}
where $l_{mn}(t) \equiv |z_m(t) - z_n(t)|$. This system is called {\it the point-vortex (PV) system}. The PV system is formally derived from the 2D Euler equations and their solutions are not equivalent to each other in general. Indeed, the derivation of the PV system is based on the Lagrangian flow map with the velocity field determined by the Biot-Savart formula, in which the velocity induced by point vortices is formally calculated. On the other hand, point vortices can approximate solutions of the 2D Euler equations. The convergence of the point-vortex approximation has been shown for smooth solutions of the 2D Euler equations in \cite{Beale,Goodman,Hald}, and weak solutions such as vortex sheets in \cite{Liu, Schochet}. Hence, analysis of the dynamics of point vortices could lead to the understanding of vortex dynamics in inviscid flows. One of the notable features of the PV system is the existence of self-similar collapsing solutions, that is, point vortices simultaneously collide with each other. The mechanism of vortex collapse plays an important role to understand fluid phenomena. For example, it has been pointed out that the vortex collapse is related to the loss of the uniqueness of solutions to the Euler equations \cite{Novikov}. The dynamics of point vortices is also utilized as a simple model of the 2D turbulence and the collapse of point vortices is considered as an elementary process in the 2D turbulence kinetics \cite{Benzi,Carnevale,Leoncini,Novikov(a),Siggia,Weiss}. Moreover, the preceding results \cite{G.,G.1} have indicated that self-similar collapse of three point vortices causes an anomalous enstrophy dissipation, which is a remarkable feature appearing in 2D turbulent flows. Although, the explicit expression of self-similar collapsing solutions have been obtained for the three PV problem in \cite{Aref,Kimura,Newton} and for an example of the four or five PV problem in \cite{Novikov}, it is still unclear what kind of configurations of point vortices lead to self-similar collapse in general. One of the main purposes of the present study is to clarify configurations of self-similar collapsing solutions and investigate the structure of the families consisting of those configurations. Another interest is to make configurations in equilibria clear. The PV system has two types of equilibria: fixed equilibria and relative equilibria. These equilibria and self-similar collapsing or rigidly translating are called {\it stationary} \cite{Oneil-1}. In this study, we focus on relative equilibria, which are rigidly rotating configurations, and try to find them as continuous limits of self-similar collapsing configurations by taking the collapse time infinity.

This paper is organized as follows. In Section~\ref{sec:PV-system}, we see some properties of the PV system and introduce the definition of self-similar motions. After showing an important proposition about the formula for self-similar solutions, we briefly review the numerical method proposed in \cite{Kudela-1, Kudela-2} to find self-similar collapsing solutions. In Section~\ref{sec:exact-solution}, we study exact solutions of self-similar collapse for $N=3$ and $N=4, 5$, which are discussed in Section~\ref{sec:three-pv} and \ref{sec:Novikov}, respectively. Through the investigation for these examples, we see that the collapse time continuously depends on the Hamiltonian and relative equilibria appears in limits of the value of the Hamiltonian. Then, in Section~\ref{sec:uniform-strength}, we study self-similar motions for the case of $N \geq 4$ under a uniform condition for vortex strengths. By numerical computations, we find families of self-similar collapsing solutions and show that symmetrical configurations in relative equilibria appear in limit states of these solutions. In particular, we give explicit configurations in relative equilibria for $N=4$ with a mathematical rigor. We also investigate an example of non-uniform vortex strengths for $N = 7$ in Section~\ref{sec:general-strength}. Section~\ref{sec:concluding} is devoted to concluding remarks.

\section{The $N$-point vortex system and self-similar solutions}
\label{sec:PV-system}
For convenience of notation, we introduce complex positions of point vortices, $z_m(t) \equiv x_m(t) + i y_m(t)$. The PV system is written by
\begin{equation}
\dfrac{\mbox{d} z_m}{\mbox{d}t} = \frac{-1}{2 \pi i} \sum_{n\neq m}^N \dfrac{\Gamma_n}{\overline{z}_m - \overline{z}_n}, \qquad z_m(0) = k_m \label{pv}
\end{equation}
for $m = 1,\cdots, N$, where $\Gamma_m$ and $k_m$ denote the strength and the initial position of $m$-th point vortex respectively, and $\overline{z}_m$ is the complex conjugate of $z_m$. The PV system has the following invariant quantities $(P, Q, I)$:
\begin{align*}
P + i Q \equiv \sum_{m=1}^N \Gamma_m x_m + i \sum_{m=1}^N \Gamma_m y_m, \quad I \equiv \sum_{m=1}^N \Gamma_m |z_m|^2 = \sum_{m=1}^N \Gamma_m (x_m^2 + y_m^2).
\end{align*}
These quantities yield another invariant $M$ depending on mutual distances of point vortices,
\begin{equation*}
M \equiv \sum_{m=1}^N \sum_{n=m+1}^N l_{mn}^2 = 2(\Gamma I - P^2 - Q^2),
\end{equation*}
where $\Gamma \equiv \sum_{m=1}^N \Gamma_m$. Considering these invariants, we find that the PV system \eqref{pv} with $N \leq 3$ is integrable for any $\Gamma_m \in \mathbb{R}$ and the system with $N = 4$ is integrable only when $\Gamma = 0$ holds, see \cite{Newton} for details. It is useful to consider the evolution of mutual distances $l_{mn}$ for $N \geq 3$, which is governed by
\begin{equation}
\frac{\mbox{d}}{\mbox{d} t} l_{mn}^2 = \frac{2}{\pi} \sum_{l \neq n \neq m}^N \Gamma_l \sigma_{mnl} A_{mnl} \left( \frac{1}{l_{nl}^2} - \frac{1}{l_{lm}^2} \right).  \label{lpv}
\end{equation}
Here, $A_{mnl}$ denotes the area of the triangle formed by $(z_m, z_n, z_l)$, which is expressed by
\begin{equation*}
A_{mnl} = \frac{1}{4} \left[ 2 \left( l_{mn}^2 l_{nl}^2 + l_{nl}^2 l_{lm}^2 + l_{lm}^2 l_{mn}^2 \right) - l_{mn}^4 - l_{nl}^4 - l_{lm}^4 \right]^{\frac{1}{2}},
\end{equation*}
and $\sigma_{mnl}$ denotes the sign of the area, that is, $\sigma_{mnl} = 1$ if the indices $(m, n, l)$ at the vertices of the triangle appear counterclockwise and $\sigma_{mnl} = -1$ if they do clockwise.

In this paper, we focus on self-similar motions of point vortices and thus assume that the complex position $z_m$ is expressed by
\begin{equation}
z_m = k_m f(t), \qquad f(t) \equiv r(t) e^{i \theta(t)}, \label{hyp-self-similar}
\end{equation}
where $r \geq 0$ and $\theta \in \mathbb{R}$ satisfy $r(0) = 1$ and $\theta(0) = 0$. Substituting \eqref{hyp-self-similar} into \eqref{pv}, we find
\begin{equation*}
2 \pi \dfrac{\mbox{d}f}{\mbox{d}t} \overline{f} = \frac{i}{k_m} \sum_{n\neq m}^N  \dfrac{\Gamma_n}{\overline{k}_m - \overline{k}_n}.
\end{equation*}
Since $f$ is independent of $m$, there exists a constant $C = A + i B \in \mathbb{C}$ with $A$, $B \in \mathbb{R}$, which is independent of $m$, such that
\begin{equation}
C = A + i B = \frac{i}{2 \pi k_m} \sum_{n\neq m}^N  \dfrac{\Gamma_n}{\overline{k}_m - \overline{k}_n} \label{def-AB}
\end{equation}
for any $m = 1,\cdots,N$. That is to say, the existence of self-similar solutions is equivalent to the existence of $\{ k_m \}_{m=1}^N$ for which there exists a constant $C \in \mathbb{C}$ satisfying \eqref{def-AB}. For the initial position $\{ k_m \}_{m=1}^N$ with a uniform constant $C = A + iB$, the self-similar solution of \eqref{pv} is explicitly described by
\begin{equation}
z_m(t) = k_m \sqrt{2At + 1} \exp{\left[ i \dfrac{B}{2A}\log{(2At + 1)} \right]} \label{sol_self-similar}
\end{equation}
for $m = 1,\cdots,N$ when $A \neq 0$. Then, the mutual distances are given by
\begin{equation*}
l_{mn}(t) = l_{mn}(0) \sqrt{2At + 1} \label{lsol_self-similar}
\end{equation*}
for $m \neq n$. Hence, the point vortices collide simultaneously at the origin and the collapse time is determined by
\begin{equation}
t_c = - \dfrac{1}{2A}. \label{t_collapse}
\end{equation}
For the initial position satisfying $A =0$, the corresponding self-similar solution is a relative equilibrium in the form,
\begin{equation*}
z_m(t) = k_m e^{i B t}.
\end{equation*}
If $B = 0$, the solution is a fixed equilibrium, which is a fixed point of the right-hand side in \eqref{pv}, and, in the other cases, the solution rotating with the angular velocity $B \neq 0$ is a fixed point of the right-hand side in \eqref{lpv}. Indeed, we have the following formula for the constants $A$ and $B$ in \eqref{def-AB}.

\begin{proposition}
\label{prop-AB}
For any self-similar solution \eqref{hyp-self-similar}, the constants $A$ and $B$ defined by \eqref{def-AB} are expressed by
\begin{align*}
A &= \frac{1}{\pi l_{mn}^2(0)} \sum_{l \neq n \neq m}^N \Gamma_l \sigma_{mnl} A_{mnl} \left( \frac{1}{l_{nl}^2} - \frac{1}{l_{lm}^2} \right), \\
B &= \frac{1}{2\pi l_{mn}^2(0)} \left[ \Gamma_m + \Gamma_n + (\boldsymbol{k}_m - \boldsymbol{k}_n) \cdot \sum_{l \neq n \neq m}^N \Gamma_l \left( \frac{\boldsymbol{k}_l - \boldsymbol{k}_n}{l_{nl}^2} - \frac{\boldsymbol{k}_l - \boldsymbol{k}_m}{l_{lm}^2} \right) \right]
\end{align*}
for any $m \neq n$, where $\boldsymbol{k}_m \equiv (Re[k_m],Im[k_m])$.
\end{proposition}

\begin{proof}
Since it follows from \eqref{def-AB} that
\begin{align*}
- 2 \pi i (k_m - k_n) C &= \sum_{l \neq m}^N \dfrac{\Gamma_l}{\overline{k}_m - \overline{k}_l} - \sum_{l \neq n}^N \dfrac{\Gamma_l}{\overline{k}_n - \overline{k}_l} \\
& =  \dfrac{\Gamma_m + \Gamma_n}{\overline{k}_m - \overline{k}_n} + \sum_{l \neq n \neq m}^N \Gamma_l \left( \dfrac{1}{\overline{k}_l - \overline{k}_n} -  \dfrac{1}{\overline{k}_l - \overline{k}_m} \right),
\end{align*}
we find
\begin{equation*}
- 2 \pi i l_{mn}^2(0) C =  \Gamma_m + \Gamma_n + \sum_{l \neq n \neq m}^N \Gamma_l \left( \dfrac{(k_l - k_n)(\overline{k}_m - \overline{k}_n)}{l_{nl}^2(0)} +  \dfrac{(k_l - k_m)(\overline{k}_n - \overline{k}_m)}{l_{lm}^2(0)} \right).
\end{equation*}
Note that, for any $\boldsymbol{x}_1 = (x_1 ,y_1)$ and $\boldsymbol{x}_2 = (x_2 ,y_2)$, we have
\begin{equation*}
(x_1 + i y_1)(x_2 - i y_2) = \boldsymbol{x}_1 \cdot \boldsymbol{x}_2 - 2 i \sigma_{12} A_{12},
\end{equation*}
where $A_{12}$ is the area of the triangle formed by the origin, $\boldsymbol{x}_1$ and $\boldsymbol{x}_2$, and $\sigma_{12}$ is the sign of the area defined in the same way as $\sigma_{mnl}$ in \eqref{lpv}. Hence, we obtain
\begin{align*}
2 \pi l_{mn}^2(0) C &= 2 \sum_{l \neq n \neq m}^N \Gamma_l \left( \dfrac{\sigma_{nlm} A_{nlm}}{l_{nl}^2(0)} + \dfrac{\sigma_{mln} A_{mln}}{l_{lm}^2(0)} \right) \\
& + i \left[ \Gamma_m + \Gamma_n + (\boldsymbol{k}_m - \boldsymbol{k}_n) \cdot \sum_{l \neq n \neq m}^N \Gamma_l \left( \dfrac{\boldsymbol{k}_l - \boldsymbol{k}_n}{l_{nl}^2(0)} - \dfrac{\boldsymbol{k}_l - \boldsymbol{k}_m}{l_{lm}^2(0)} \right) \right].
\end{align*}
Considering $\sigma_{mnl} = \sigma_{nlm} = -\sigma_{mln}$, we have the desired result.
\end{proof}

For any self-similar collapsing motion, its reflection over the $x$-axis or the $y$-axis leads to a self-similar expanding motion. Indeed, Proposition~\ref{prop-AB} implies that the signs of all $\sigma_{mnl}$ get the opposite signs by the reflection and thus the sign of $A$ changes, while the constant $B$ and the absolute value of $A$ are invariant under reflections. In addition, the change of the signs for all $\sigma_{mnl}$ is equivalent to the case that the signs of all vortex strengths get reversed, for which the orbits of point vortices coincide with those before the reflection for the negative time direction \cite{Synge}. We remark that, for the self-similar solutions \eqref{hyp-self-similar}, the invariant quantities need to satisfy $P = Q = I = M = 0$, and the invariance of the Hamiltonian gives
\begin{equation}
\Gamma_H \equiv \sum_{m=1}^N \sum_{n=m+1}^N \Gamma_m \Gamma_n = 0. \label{condi-gmn}
\end{equation}
The center of vorticity $(P + i Q)/\Gamma$, at which the collapse of point vortices occurs, is also invariant and located at origin for $\Gamma \neq 0$.

As we see in Section~\ref{sec:three-pv}, it is well known that the three PV system has a necessary and sufficient condition for the existence of self-similar collapsing solutions. For the case of $N \geq 4$, although an example of exact self-similar collapsing solutions has been obtained for the four and five PV system \cite{Novikov}, which is treated in Section~\ref{sec:Novikov}, it has not been made clear what kind of configurations lead to self-similar collapse. On the other hand, in the recent studies by \cite{Kudela-1,Kudela-2}, several examples of self-similar collapsing solutions have been constructed numerically for $N \geq 4$. In this paper, we use the method proposed in those papers to derive configurations leading to self-similar collapse. We now review that method briefly. Let the complex variable $v_m$ be
\begin{equation*}
v_m \equiv \frac{-1}{2 \pi i} \sum_{n\neq m}^N \dfrac{\Gamma_n}{\overline{k}_m - \overline{k}_n}.
\end{equation*}
Then, the existence of the constant $C$ satisfying \eqref{def-AB} is equivalent to the existence of $\{ k_m \}_{m=1}^N$ such that $v_m k_l = v_l k_m$ for $m \neq l = 1,\cdots N$, which is equivalent to
\begin{equation}
v_1 k_m = v_m k_1, \label{v1_km}
\end{equation}
for $m = 2, \cdots, N-2$, see \cite{Kudela-1,Kudela-2,Oneil-1}. Since self-similar solutions are invariant under self-similar transformations, we can fix the position of $N$-th point vortex. To determine the positions of the rest of $N-1$ point vortices, $2N-2$ equations for $\{ k_m \}$ are required. Thus, we use the following nonlinear equations:
\begin{align}
f_{2m-3} \equiv Re[v_1 k_m] - Re[v_m k_2] = 0, \quad f_{2m-2} \equiv Im[v_1 k_m] - Im[v_m k_2] = 0 \label{f-1}
\end{align}
for $m = 2,\cdots N-2$, and
\begin{align}
f_{2N-5} \equiv P = 0, \quad f_{2N-4} \equiv Q = 0, \quad f_{2N-3} \equiv S = 0, \quad f_{2N-2} \equiv \mathscr{H} - \mathscr{H}_0 = 0, \label{f-2}
\end{align}
where $\mathscr{H}_0$ is a given constant. As explained in \cite{Kudela-1,Kudela-2}, these equations can be solved numerically by the Newton method with a starting point determined by applying the Levenberg-Marquart algorithm to the nonlinear least squares problem $\sum_{n=1}^{2N-2}f_n^2$. From the above method, for the given constant $\mathscr{H}_0$, we can numerically obtain a configuration leading to a self-similar motion if a self-similar solution whose Hamiltonian is $\mathscr{H}_0$ exists.

\section{Exact solutions for self-similar collapse}
\label{sec:exact-solution}
\subsection{Three point-vortex problem}
\label{sec:three-pv}

In this section, we see the properties of self-similar collapsing solutions of the three PV problem. It is well known that the conditions,
\begin{equation}
\Gamma_H = \Gamma_1 \Gamma_2 + \Gamma_2 \Gamma_3 + \Gamma_3 \Gamma_1 = 0,  \label{three-g-condi}
\end{equation}
\begin{equation}
M = \Gamma_1 \Gamma_2 l_{12}^2 + \Gamma_2 \Gamma_3 l_{23}^2 + \Gamma_3 \Gamma_1 l_{31}^2 = 0, \label{M_pv}
\end{equation}
are a necessary and sufficient condition for self-similar collapse of three point vortices \cite{Aref,Kimura,Newton}, and partial collapse does not occur. Note that (\ref{three-g-condi}) allows us to assume $\Gamma_1 \geq \Gamma_2 > 0 > \Gamma_3$ without loss of generality and replace $\Gamma_3$ by $ - \Gamma_1 \Gamma_2 /(\Gamma_1 + \Gamma_2)$. Under the conditions \eqref{three-g-condi} and \eqref{M_pv}, we have a self-similar solution expressed by \eqref{sol_self-similar} with
\begin{align}
A &= \dfrac{\Gamma_3}{4 \pi l_{12}^2(0)} \sigma_{123} \left[ 2(\lambda_1 + \lambda_2) - (\lambda_1 - \lambda_2)^2 -1 \right]^{1/2} \dfrac{\lambda_2 - \lambda_1}{\lambda_1\lambda_2}, \label{three-A} \\
B &= \dfrac{1}{4 \pi l_{12}^2(0)} \dfrac{\lambda_1\lambda_2(\Gamma_1 + \Gamma_2)^3 + (\Gamma_1^2 + \Gamma_2^2)(\Gamma_2 \lambda_1 + \Gamma_1 \lambda_2)}{\lambda_1\lambda_2(\Gamma_1 + \Gamma_2)^2}, \nonumber
\end{align}
where $\lambda_1 \equiv l_{23}^2 / l_{12}^2$ and $\lambda_2 \equiv l_{31}^2 / l_{12}^2$ are invariant constants \cite{Newton}. The constant $A$ gives the collapse time $t_c$ by the formula \eqref{t_collapse}. All possible equilibria satisfying \eqref{three-g-condi} and \eqref{M_pv} are collinear states or equilateral triangles: they both form relative equilibria rotating rigidly about their center of vorticity and, in particular, the frequency of equilateral triangle is given by $\Gamma / (2 \pi l_{mn}^2)$, see \cite{Newton} for details. According to \cite{G.,Novikov}, the Hamiltonian $\mathscr{H}$ of the self-similar collapsing solution satisfies
\begin{equation}
\min{\{\mathscr{H}_+, \mathscr{H}_-\}} < \mathscr{H} < 0, \label{Range-H}
\end{equation}
where
\begin{equation*}
\mathscr{H}_\pm \equiv \frac{\Gamma_1^2 \Gamma_2^2}{4\pi(\Gamma_1 + \Gamma_2)} \log\left( \psi\left( \frac{\Gamma_1}{\Gamma_2}k_{\pm} \right) \left[\psi\left( \frac{\Gamma_1}{\Gamma_2} \right)\right]^{-1}  \right)
\end{equation*}
and $\psi(r)$ and $k_\pm$ are given by
\begin{equation*}
\psi(r) \equiv \left( \frac{1}{1+r} \right)^{1/\Gamma_1}  \left( \frac{r}{1+r} \right)^{1/\Gamma_2}, \quad k_\pm \equiv \left( \frac{\Gamma_1 + \Gamma_2 \pm \sqrt{\Gamma_1^2 + \Gamma_1 \Gamma_2 + \Gamma_2^2}}{\Gamma_2} \right)^2.
\end{equation*}
Then, $\mathscr{H} = \mathscr{H}_\pm$ and $\mathscr{H} = 0$ correspond to collinear states and an equilateral triangle in relative equilibria, respectively.

\begin{figure}[t]
\begin{center}
\includegraphics[scale=0.65]{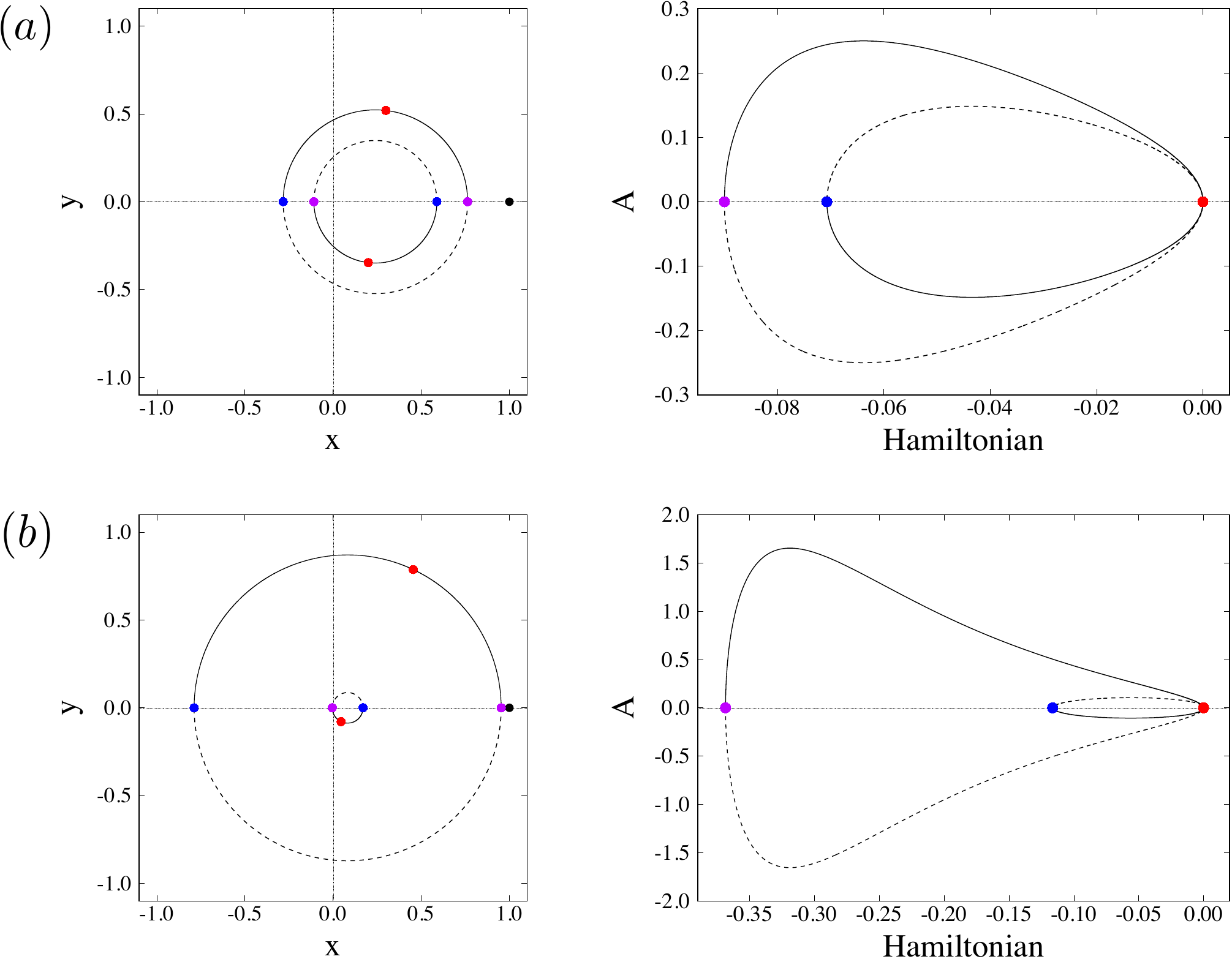}
\end{center}
\caption{The graphs of \eqref{three-init-general} and $(\mathscr{H}(\theta), A(\theta))$ by changing $0 \leq \theta < 2 \pi$ for $(a)$ $(\Gamma_1, \Gamma_2) = (3/2, 1)$ and $(b)$ $(\Gamma_1, \Gamma_2) = (10, 1)$. The solid lines denotes the graphs for $0 \leq \theta < \pi$ and the dashed ones are for $\pi \leq \theta < 2 \pi$. The black point in the left figure denotes the fixed $k_3$. The configuration formed by red points with $k_3$ is an equilateral triangle for $\theta = \theta_0$, and those by blue or purple points with $k_3$ are collinear states for $\theta = 0$ or $\pi$, respectively. }
\label{fig:three_general}
\end{figure}

We see the intermediate configurations between collinear states and equilateral triangles: they lead to self-similar collapse and satisfy \eqref{Range-H}. Setting $k_1 = x_1$ and $k_2 = x_2$ with $x_1 > x_2$, we find from \eqref{M_pv} that $k_3 = x + iy$ with $x$, $y\in\mathbb{R}$ satisfies
\begin{equation*}
\left(x - \dfrac{\Gamma_1 x_1 + \Gamma_2 x_2}{\Gamma_1 + \Gamma_2} \right)^2 + y^2 = \dfrac{\Gamma_1^2 + \Gamma_1 \Gamma_2 + \Gamma_2^2}{(\Gamma_1 + \Gamma_2)^2} |x_1 - x_2|^2.
\end{equation*}
Thus, we have
\begin{equation*}
x = \dfrac{\Gamma_1 x_1 + \Gamma_2 x_2}{\Gamma_1 + \Gamma_2} + \dfrac{\sqrt{\mathcal{R}}}{\Gamma_1 + \Gamma_2} l \cos{\theta}, \qquad y = \dfrac{\sqrt{\mathcal{R}}}{\Gamma_1 + \Gamma_2} l \sin{\theta}
\end{equation*}
for $\theta \in [0, 2\pi)$, where $\mathcal{R} \equiv \Gamma_1^2 + \Gamma_1 \Gamma_2 + \Gamma_2^2$ and $l \equiv |x_1 - x_2|$, see \cite{Kimura}. Then, their mutual distances are given by
\begin{equation}
l_{23} = \dfrac{l}{\Gamma_1 + \Gamma_2} \sqrt{ \Gamma_1^2 + \mathcal{R} + 2 \Gamma_1 \sqrt{\mathcal{R}} \cos{\theta} }, \quad l_{31} = \dfrac{l}{\Gamma_1 + \Gamma_2} \sqrt{\Gamma_2^2 + \mathcal{R} - 2 \Gamma_2 \sqrt{\mathcal{R}} \cos{\theta} }  \label{three-l-theta}
\end{equation}
and $l_{12} = l$. The formula \eqref{three-l-theta} implies that three point vortices form an equilateral triangle when $\theta = \theta_0$ such that $\cos{\theta_0} = - (\Gamma_1 - \Gamma_2)/(2 \sqrt{\mathcal{R}})$. To fix $k_3$ to $(1,0)$, we apply a self-similar transformation to $(k_1, k_2, k_3)$ so that
\begin{equation}
k_1 = \dfrac{\Gamma_1 \Gamma_2}{(\Gamma_1 + \Gamma_2)^2} \left( 1 + \dfrac{\sqrt{\mathcal{R}}}{\Gamma_1} e^{- i \theta} \right), \quad k_2 = \dfrac{\Gamma_1 \Gamma_2}{(\Gamma_1 + \Gamma_2)^2} \left( 1 - \dfrac{\sqrt{\mathcal{R}}}{\Gamma_2} e^{- i \theta} \right), \quad k_3 = 1. \label{three-init-general}
\end{equation}
Then, $(k_1, k_2, k_3)$ satisfy $P = Q = 0$ and their mutual distances are given by \eqref{three-l-theta} with $l_{12} = \sqrt{\mathcal{R}}/(\Gamma_1 + \Gamma_2)$. Thus, the Hamiltonian is expressed by a function of $\theta$, that is,
\begin{equation*}
\mathscr{H}(\theta) = \dfrac{\Gamma_1 \Gamma_2}{4\pi (\Gamma_1 + \Gamma_2)} \log{\left( \dfrac{\Gamma_1^2 + \mathcal{R} + 2 \Gamma_1 \sqrt{\mathcal{R}} \cos{\theta}}{(\Gamma_1 + \Gamma_2)^2} \right)^{\Gamma_2} \left( \dfrac{\Gamma_2^2 + \mathcal{R} - 2 \Gamma_2 \sqrt{\mathcal{R}} \cos{\theta}}{(\Gamma_1 + \Gamma_2)^2} \right)^{\Gamma_1}},
\end{equation*}
and $\mathscr{H}(\theta)$ is invariant under self-similar transformations. Note that $A$ in \eqref{three-A} is also considered as a function of $\theta$ by using \eqref{three-l-theta} and we describe it by $A(\theta)$. Figure~\ref{fig:three_general} shows the sets of positions of three point vortices \eqref{three-init-general} and the dependence between $A(\theta)$ and $\mathscr{H}(\theta)$, which are obtained by changing $\theta \in [0,2\pi)$. Since the graphs for $\theta \in (\pi, 2 \pi)$ in both figures are symmetric to those for $\theta \in (0, \pi)$ with the $x$-axis or the $\mathscr{H}$-axis, we pay attention only to $\theta \in [0, \pi]$. Note that $A = 0$ holds for two collinear states for $\theta = 0$, $\pi$, and an equilateral triangle for $\theta = \theta_0$. As we see in the right figure, these relative equilibria are continuously connected by a family of self-similar solutions satisfying $A \neq 0$: solutions for $0 < \theta < \theta_0$ are self-similar collapsing with $A < 0$, and solutions for $\theta_0 < \theta < \pi$ are self-similar expanding with $A > 0$. According to \eqref{Range-H}, the equilateral triangle for $\theta = \theta_0$ satisfies $\mathscr{H} = 0$ and the collinear states for $\theta = 0$ and $\pi$ do $\mathscr{H} = \mathscr{H}_-$ and $\mathscr{H}_+$, respectively.

\begin{figure}[t]
\begin{center}
\includegraphics[scale=0.62]{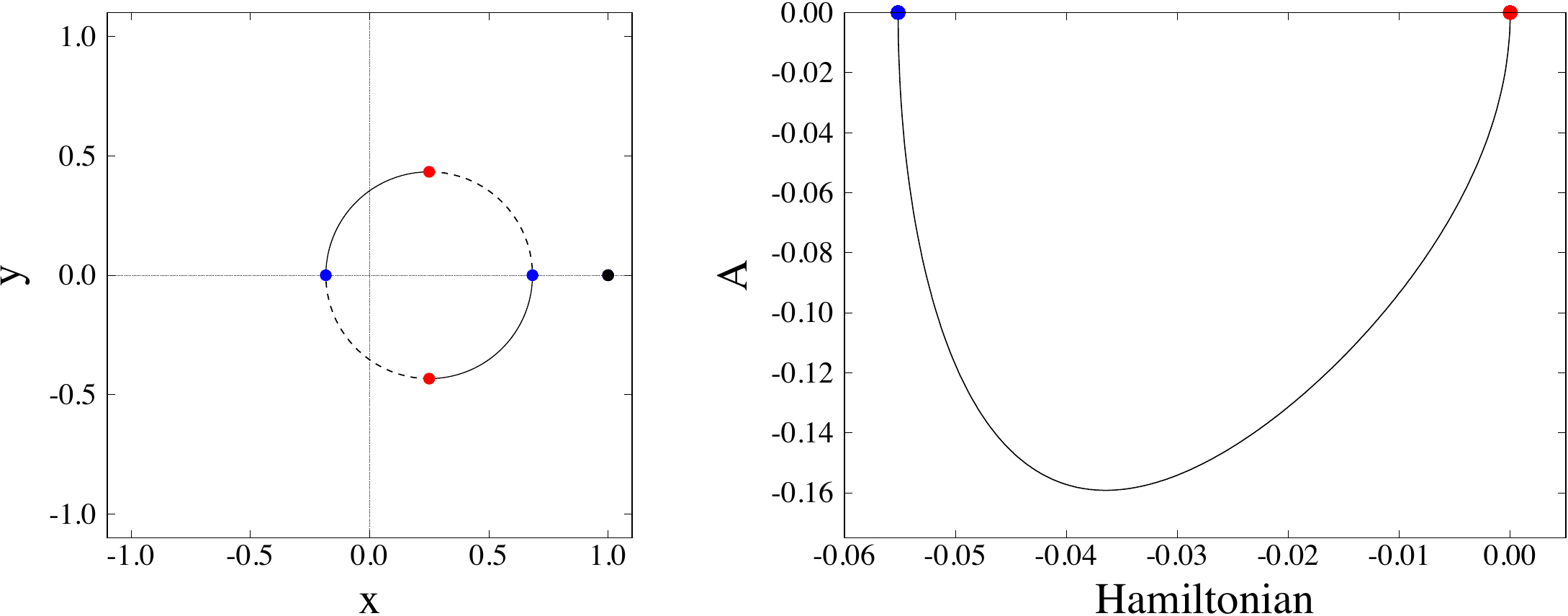}
\caption{The graphs of \eqref{ex-three} and $(\mathscr{H}(\theta), A(\theta))$ in \eqref{three-ABH} by changing $0 \leq \theta \leq \pi/2$. In the left figure, the dashed lines express the configurations for $\pi/2 < \theta < \pi$ and the black point denotes $k_3$. The configurations formed by red and blue points correspond to an equilateral triangle for $\theta = \pi /2$ and a collinear state for $\theta = 0$, respectively.}
\label{fig:three}
\end{center}
\end{figure}

For the case of $\Gamma_1 = \Gamma_2$, owing to \eqref{three-g-condi}, we may assume that $\Gamma_1 = \Gamma_2 = 1$ and $\Gamma_3 = -1/2$ without loss of generality. Then, the initial configuration \eqref{three-init-general} is written by
\begin{equation}
k_1 = \dfrac{1}{4}\left( 1 + \sqrt{3} e^{-i\theta} \right), \qquad  k_2 = \dfrac{1}{4}\left( 1 - \sqrt{3} e^{-i\theta} \right), \qquad k_3 = 1, \label{ex-three}
\end{equation}
which is a modification of the configuration given by \cite{Kimura}. For the same reason as the case of general vortex strengths, we restrict the domain of $\theta$ to $[0,\pi]$. Moreover, under the condition $\Gamma_1 = \Gamma_2$, $k_2$ can be identified with $k_1$ and configurations for $(\pi/2, \pi)$ coincide with those for $(0,\pi/2)$ by a reflection over the $x$-axis and replacing $k_1$ and $k_2$ each other. Hence, it is enough to pay attention to $\theta \in [0,\pi/2]$. Note that the equilateral triangle and the collinear state are equivalent to $\theta = \pi /2$ and $\theta = 0$, respectively. Then, the formulae of $A$, $B$ and $\mathscr{H}$ for \eqref{ex-three} are simplified as follows.
\begin{equation}
A(\theta) = \dfrac{- 2 \sin{2\theta}}{\pi(5 -3 \cos{2\theta})}, \quad B(\theta) = \dfrac{2(3 - \cos{2\theta})}{\pi(5 -3 \cos{2\theta})}, \quad \mathscr{H}(\theta) = \dfrac{1}{8 \pi} \log{\dfrac{5 -3 \cos{2\theta}}{8}} \label{three-ABH}
\end{equation}
for $0 \leq \theta \leq \pi/2$, which give $\mathscr{H}(\pi/2) = 0$ and $\mathscr{H}_\pm = \mathscr{H}(0) = - 0.0551589\cdots$. Figure~\ref{fig:three} shows the set of \eqref{ex-three} and the dependence between $A$ and $\mathscr{H}$ in \eqref{three-ABH} by changing $\theta \in [0,\pi/2]$. Note that, in the right figure, the graph corresponding to $\theta \in [\pi/2, \pi]$ is symmetric to that for $\theta \in [0 ,\pi/2]$ with the $\mathscr{H}$-axis, that is, the curve satisfying $A > 0$ for $\mathscr{H}_\pm < \mathscr{H} < 0$ and connected to $(\mathscr{H}_\pm, 0)$ and $(0, 0)$ in the $\mathscr{H}$-$A$ plane.

\subsection{Four and five point-vortex problems}
\label{sec:Novikov}

For the four and five point-vortex problem, a necessary and sufficient condition has not been established. On the other hand, an example of exact solutions for self-similar collapse has been obtained by \cite{Novikov}. We see more detail about this example. Consider the four point vortices located at the vertices of a parallelogram whose diagonals intersect at the origin. Without loss of generality, we can set
\begin{equation}
k_1 = \dfrac{1}{2} d_1 e^{i\theta}, \qquad k_2 = - \dfrac{1}{2} d_1 e^{i\theta}, \qquad k_3 = - \dfrac{1}{2} d_2, \qquad k_4 = \dfrac{1}{2} d_2, \label{ex-four}
\end{equation}
where $l_{12} = d_1$ and $l_{34}= d_2$ are the diagonals and $\theta$ is the angle between the diagonals. The strengths of point vortices are given by $\Gamma_1 = \Gamma_2 = \gamma_1$ and $\Gamma_3 = \Gamma_4 = \gamma_2$, where $\gamma_1$, $\gamma_2 \in \mathbb{R}$ are given constants. Note that since $k_2$ is identified with $k_1$, we restrict the domain of $\theta$ to $[0, \pi/2]$ with the same reason as the example of three point vortices \eqref{ex-three}. Then, we easily find that $P = Q = 0$ holds and the following conditions for self-similar motions are required.
\begin{align*}
& I = \dfrac{1}{2}\left( \gamma_1 d_1^2 + \gamma_2 d_2^2 \right) = 0,\\
& M = (\gamma_1 + \gamma_2)\left( \gamma_1 d_1^2 + \gamma_2 d_2^2 \right) = 0, \\
& \Gamma_H = \gamma_1^2 + 4 \gamma_1 \gamma_2 + \gamma_2^2 = 0.
\end{align*}
These conditions yield the relation,
\begin{equation}
\dfrac{d_1^2}{d_2^2} = - \dfrac{\gamma_2}{\gamma_1} = 2 \pm \sqrt{3}, \label{condi-four}
\end{equation}
which implies that $\gamma_1$ and $\gamma_2$ have opposite signs. For the five PV problem, the fifth point vortex is placed at the origin, i.e. $k_5 = 0$, with the strength $\Gamma_5 = \gamma_3 \in \mathbb{R}$. Similarly to the four PV problem, it is confirmed that $P = Q = 0$ and
\begin{align}
& I = \dfrac{1}{2}\left( \gamma_1 d_1^2 + \gamma_2 d_2^2 \right) = 0, \label{condi-I-five}\\
& M = \dfrac{1}{2}(2 \gamma_1 + 2 \gamma_2 + \gamma_3)\left( \gamma_1 d_1^2 + \gamma_2 d_2^2 \right) = 0,  \nonumber \\
& \Gamma_H = \gamma_1^2 + 4 \gamma_1 \gamma_2 + \gamma_2^2 + 2 \gamma_3 (\gamma_1 + \gamma_2)= 0. \label{condi-g-five}
\end{align}

\begin{figure}[t]
\begin{center}
\includegraphics[scale=0.62]{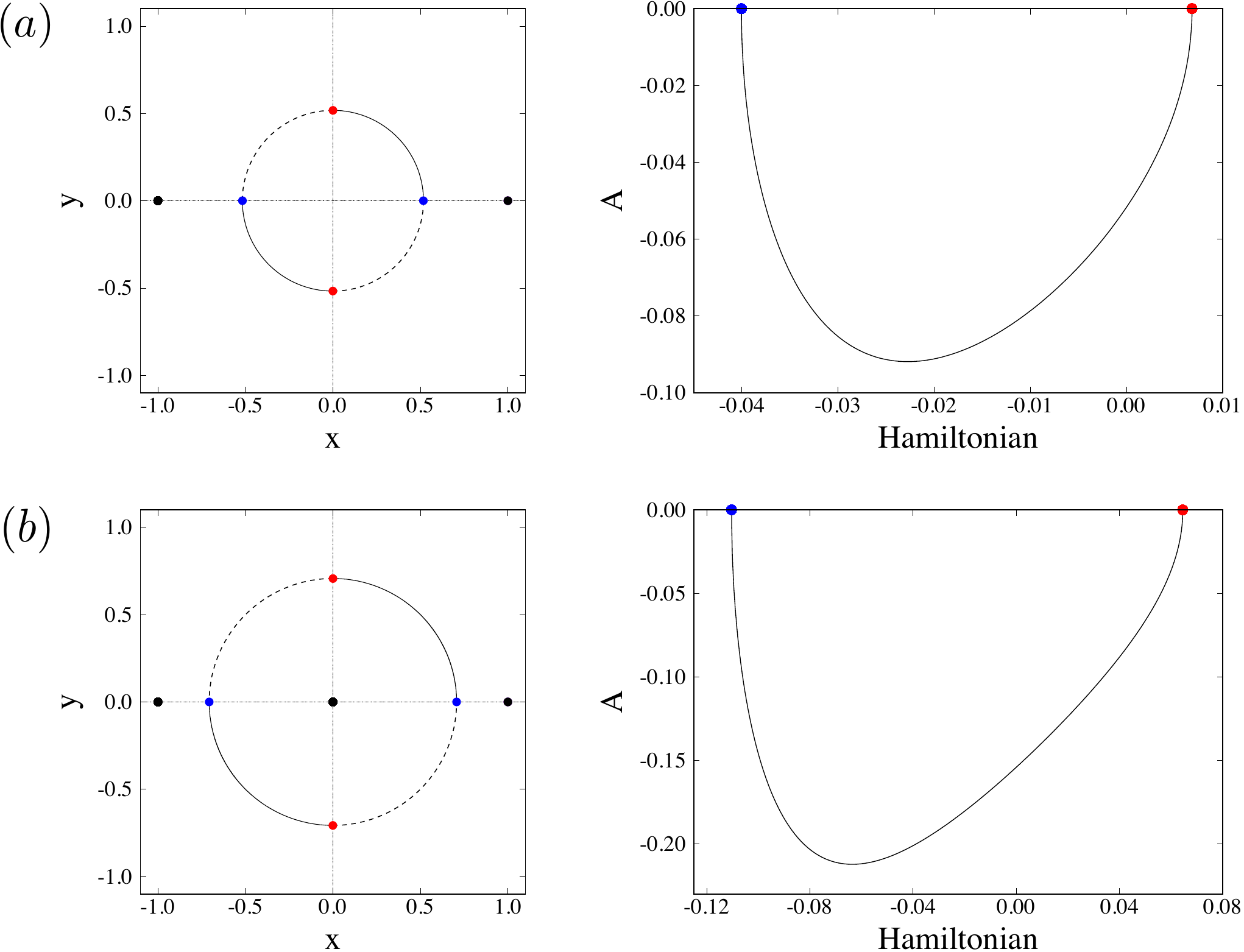}
\caption{$(a)$ The graphs of \eqref{ex-four} and $(\mathscr{H}(\theta), A(\theta))$ for $0 \leq \theta \leq \pi/2$, in which $d_2 = 2$, $\gamma_1 = -1$, $\gamma_3 =0$ and $d_1$, $\gamma_2$ are determined by \eqref{condi-four} with the ratio $2 - \sqrt{3}$. $(b)$ The graphs of \eqref{ex-four} with $k_5 = 0$ and $(\mathscr{H}(\theta), A(\theta))$ for $0 \leq \theta \leq \pi/2$, in which $d_2 = 2$, $\gamma_1 = -1$, $\gamma_2 = 1/2$ and $d_1$, $\gamma_3$ are determined by \eqref{condi-I-five} and \eqref{condi-g-five}. The black points in the left figures denote the fixed point vortices. The red and blue points correspond to $\theta = \pi /2$ and $\theta = 0$, respectively.}
\label{fig:Novikov}
\end{center}
\end{figure}

We introduce several formulae for the above example of five point vortices and their formulae are valid for the four point vortices \eqref{ex-four} by setting $\gamma_3 = 0$. According to \cite{Novikov}, for the configuration formed by \eqref{ex-four} with $k_5 = 0$, the constants $A$ and $B$ in \eqref{def-AB} are calculated by
\begin{align}
A(\theta) &= \dfrac{4 \gamma_1 d_1^2 \sin{2 \theta}}{\pi \left(d_1^4 + d_2^4 -2 d_1^2 d_2^2 \cos{2 \theta}  \right)}, \label{four-A}\\
B(\theta) &= \dfrac{d_1^2 + d_2^2}{2 \pi d_1^2 d_2^2} \left[ \gamma_1 + \gamma_2 + 2 \gamma_3 + \dfrac{4 \gamma_2 d_2^2(d_1^2 - d_2^2)}{d_1^4 + d_2^4 -2 d_1^2 d_2^2 \cos{2 \theta}} \right], \nonumber
\end{align}
which can be derived based on Proposition~\ref{prop-AB}, and the Hamiltonian $\mathscr{H}$ is expressed by
\begin{equation}
\mathscr{H}(\theta) = \dfrac{-1}{2 \pi} \log{\left[ c_\gamma d_1^{\gamma_1 (\gamma_1 + 2 \gamma_3)} d_2^{\gamma_2 (\gamma_2 + 2 \gamma_3)} \left(d_1^4 + d_2^4 - 2 d_1^2 d_2^2 \cos{2 \theta} \right)^{\gamma_1 \gamma_2} \right]}, \label{four-H}
\end{equation}
where $c_\gamma \equiv 2^{ -4 \gamma_1 \gamma_2 -2 \gamma_3 (\gamma_1 + \gamma_2)}$.

Figure~\ref{fig:Novikov}(a) shows the graphs of \eqref{ex-four} and the dependence between \eqref{four-A} and \eqref{four-H} with $\gamma_3 = 0$ by changing $\theta\in [0, \pi/2]$, in which we set $d_2 = 2$, $\gamma_1 = -1$ and $d_1$, $\gamma_2$ are determined by \eqref{condi-four} with the ratio $2 - \sqrt{3}$. The configurations for $\theta = \pi/2$ and $\theta = 0$ correspond to a diamond configuration and a collinear state in relative equilibria, respectively. Figure~\ref{fig:Novikov}(b) shows the graphs of \eqref{ex-four} with $k_5 =0$ and the dependence between \eqref{four-A} and \eqref{four-H} for $\theta\in [0, \pi/2]$, in which we set $d_2 = 2$, $\gamma_1 = -1$, $\gamma_2 = 1/2$. The constants $d_1$ and $\gamma_3$ are determined by \eqref{condi-I-five} and \eqref{condi-g-five}, respectively. The configurations for $\theta = \pi/2$ and $\theta = 0$ correspond to a diamond configuration and a collinear state with $k_5$ fixed at the origin in relative equilibria, respectively. As we see in the right figures in both examples, these relative equilibria are connected by the family of self-similar collapsing solutions. Similarly to the example of the three point vortices \eqref{ex-three}, the graph corresponding to $\theta \in [\pi/2, \pi]$ in the right figure satisfies $A > 0$ and it is symmetric to that for $\theta \in [0 ,\pi/2]$ with the $\mathscr{H}$-axis.

\section{Self-similar motions of $N$-point vortices with $N \geq 4$}
\label{sec:N-pv}
\subsection{Motions of point vortices with uniform strengths}
\label{sec:uniform-strength}

In this section, we investigate self-similar collapse of $N$-point vortices for $N \geq 4$. To compare with the three PV problem in Section~\ref{sec:three-pv}, we consider the following strengths of point vortices.
\begin{equation}
\Gamma_1 = \Gamma_2 = \cdots = \Gamma_{N-1} = 1, \qquad \Gamma_N = -\dfrac{N-2}{2}, \label{condi_strength}
\end{equation}
which satisfies the condition \eqref{condi-gmn}. Considering the invariance of self-similar motions under scaling and rotating transformations, we may fix the position of the $N$-th point vortex to $(1,0)$. Our concern is the family of configurations leading to self-similar collapse, which we call {\it the collapsing family}. More precisely, we try to obtain the family in which configurations are continuously connected by considering the value of Hamiltonian as a parameter. For the exact solutions shown in Section~\ref{sec:exact-solution}, the constants $A$ and $\mathscr{H}$ are parametrized by $\theta$ and the dependence between them is explicitly described by a curve on the $\mathscr{H}$-$A$ plane. Although, in general, it is difficult to represent $A$ and $\mathscr{H}$ by functions of one parameter, it is possible to obtain the curves describing the dependence between $A$ and $\mathscr{H}$ by numerical computation. For later use, we call those curves {\it the $\mathscr{H}$-$A$ curves}. From the investigations for the exact self-similar collapsing solutions, we expect that $\mathscr{H}$-$A$ curves are continuously connected to configurations in relative equilibria. Thus, finding collapsing families, we can make configurations in relative equilibria clear as the end points of $\mathscr{H}$-$A$ curves, which satisfy $A = 0$.

In order to carry out the above investigation, we employ the numerical method proposed in \cite{Kudela-1,Kudela-2}. According to that method, once we find a configuration leading to self-similar collapse for $\mathscr{H} = \mathscr{H}_0$, using its configuration as a new starting point, we can detect another configuration of a self-similar collapsing solution by numerically solving \eqref{f-1} and \eqref{f-2} with the perturbed values of the Hamiltonian $\mathscr{H}_0 \pm \Delta \mathscr{H}$, that is, $f_{2N-2} = \mathscr{H} - (\mathscr{H}_0 \pm \Delta \mathscr{H})$. Repeating this procedure by using the perturbed Hamiltonian and taking $|\Delta \mathscr{H}|$ small when it failed to find a configuration by solving \eqref{f-1} and \eqref{f-2}, we obtain configurations whose constant $A$ continuously depends on $\mathscr{H}$.

\begin{figure}[t]
\begin{center}
\includegraphics[scale=0.6]{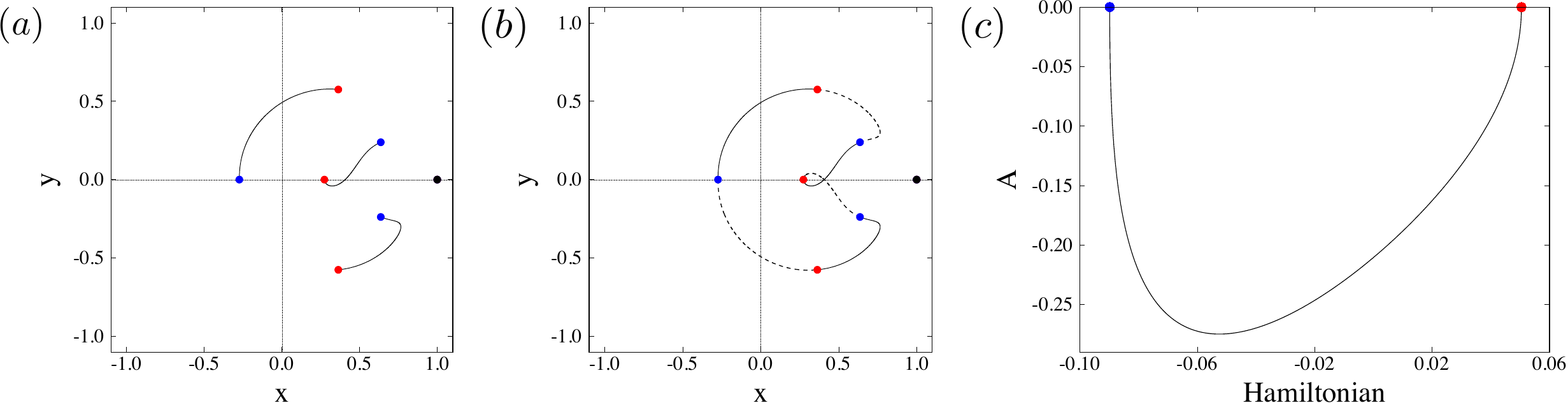}
\caption{$(a)$ The collapsing family for $N = 4$ with \eqref{condi_strength}. $(b)$ The solid curve is the same as $(a)$ and the dashed one is symmetric to $(a)$ with $x$-axis. The black points denote the fixed $k_4$. The configurations formed by red or blue points with $k_4$ are in relative equilibria. $(c)$ The $\mathscr{H}$-$A$ curve for $N = 4$. The red or blue points correspond to the relative equilibria.}
\label{fig:four}
\end{center}
\end{figure}

We first see the case of $N = 4$. Applying the numerical method described above, we obtain a collapsing family and a $\mathscr{H}$-$A$ curve, see Figure~\ref{fig:four}$(a)$ and $(c)$, respectively. Unlike the examples in Section~\ref{sec:exact-solution}, we observe that the collapsing family does not form parts of a circle and it is difficult to give an explicit formula for configurations leading to self-similar collapse. Note that the symmetric configurations to the collapsing family with the $x$-axis yield a family consisting of self-similar expanding solutions along the positive time direction, which we call {\it the expanding family}. Indeed, as we mentioned below Proposition~\ref{prop-AB}, the sign of $A$ changes by the reflection over the $x$-axis and, thus, the $\mathscr{H}$-$A$ curve for the expanding family is symmetric to the curve for the collapsing family with the $\mathscr{H}$-axis. Figure~\ref{fig:four}$(b)$ shows the collapsing and expanding families. As we see in that figure, the expanding family is connected to the collapsing family on the $x$-$y$ plane at the configurations in relative equilibria and two families form a closed curve. On the other hand, similarly to the examples in Section~\ref{sec:exact-solution}, the end points of the $\mathscr{H}$-$A$ curves correspond to relative equilibria since we have $A = 0$ at those points. Two configurations in relative equilibria seem to be symmetric with respect to the $x$-axis and this fact is shown with a mathematical rigor as follows.

\begin{proposition}
\label{four-equilibria}
For the four PV system with the condition \eqref{condi_strength}, the configurations formed by $k_1 = a$, $k_2 = b + i c$, $k_3 = b - ic$ and $k_4 = 1$, where
\begin{equation*}
a = \mp \dfrac{1}{3} \left(\sqrt{5} - \sqrt{2}\right), \quad b = \dfrac{1}{2} \pm \dfrac{1}{6} \left(\sqrt{5} - \sqrt{2}\right), \quad c = \sqrt{\dfrac{1}{6}\left(\sqrt{10} -2 \mp \left(\sqrt{5} - \sqrt{2}\right) \right)},
\end{equation*}
are in relative equilibria.
\end{proposition}

\begin{proof}
We determine the constants $(a,b,c)$ so that $P = 0$, $Q = 0$, $I = 0$ and $A = 0$ are satisfied. Note that symmetric configurations about the $x$-axis satisfy $Q = 0$ and
\begin{equation*}
Re \left[ \dfrac{i}{k_1} \sum_{n \neq 1}^4 \dfrac{\Gamma_n}{\overline{k}_1 - \overline{k}_n} \right] = Re \left[ \dfrac{i}{k_4} \sum_{n\neq 4}^4 \dfrac{\Gamma_n}{\overline{k}_1 - \overline{k}_n} \right] = 0.
\end{equation*}
On the other hand, the conditions $P = 0$ and $I = 0$ yield
\begin{equation}
a + 2b - 1 = 0, \qquad a^2 + 2 b^2 + 2 c^2 - 1 = 0, \label{abc-1}
\end{equation}
respectively. Note that $0 < b < 2/3$ follows from \eqref{abc-1}. In addition, we have
\begin{align*}
& Re \left[ \dfrac{i}{k_2} \sum_{n\neq 2}^4 \dfrac{\Gamma_n}{\overline{k}_1 - \overline{k}_n} \right] = - Re \left[ \dfrac{i}{k_3} \sum_{n\neq 3}^4 \dfrac{\Gamma_n}{\overline{k}_1 - \overline{k}_n} \right] \\
& = -\dfrac{c}{|k_2|^2} \left[ \dfrac{b}{2 c^2} + \dfrac{a}{(b-a)^2 + c^2} - \dfrac{1}{(b-1)^2 + c^2} \right],
\end{align*}
and thus $A = 0$ is equivalent to
\begin{equation}
\dfrac{b}{2 c^2} + \dfrac{a}{(b-a)^2 + c^2} - \dfrac{1}{(b-1)^2 + c^2} = 0. \label{abc-2}
\end{equation}
Combining \eqref{abc-1} and \eqref{abc-2}, we find
\begin{equation*}
36 b^4 - 72 b^3 + 40 b^2 - 4 b - 1 = 0,
\end{equation*}
which is symmetric about $b = 1/2$. Solving this equation for $0 < b < 2/3$, we obtain
\begin{equation*}
b = \dfrac{1}{2} \pm \dfrac{1}{6} \left(\sqrt{5} - \sqrt{2}\right),
\end{equation*}
and the rest of constants are calculated by \eqref{abc-1}:
\begin{equation*}
a = \mp \dfrac{1}{3} \left(\sqrt{5} - \sqrt{2}\right), \qquad c = \sqrt{\dfrac{1}{6}\left(\sqrt{10} -2 \mp \left(\sqrt{5} - \sqrt{2}\right) \right)},
\end{equation*}
where the double-sign corresponds.
\end{proof}

\begin{figure}[pt]
\begin{center}
\includegraphics[scale=0.6]{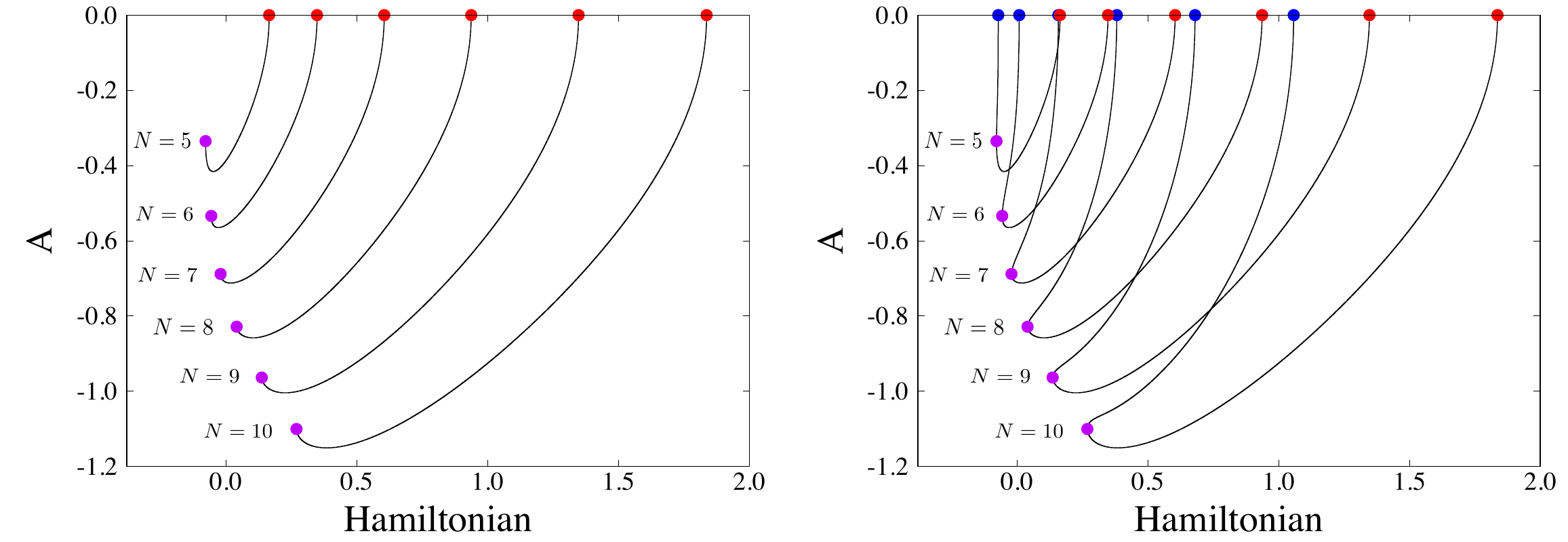}
\caption{The $\mathscr{H}$-$A$ curves for $N = 5,\cdots,10$. The red and blue points correspond to relative equilibria. The purple points correspond to $(\mathscr{H}_c, A_c)$ at which the values of Hamiltonian along the curves are minimum.}
\label{fig:H-A-fiv-ten}

\vspace{10mm}

\includegraphics[scale=0.6]{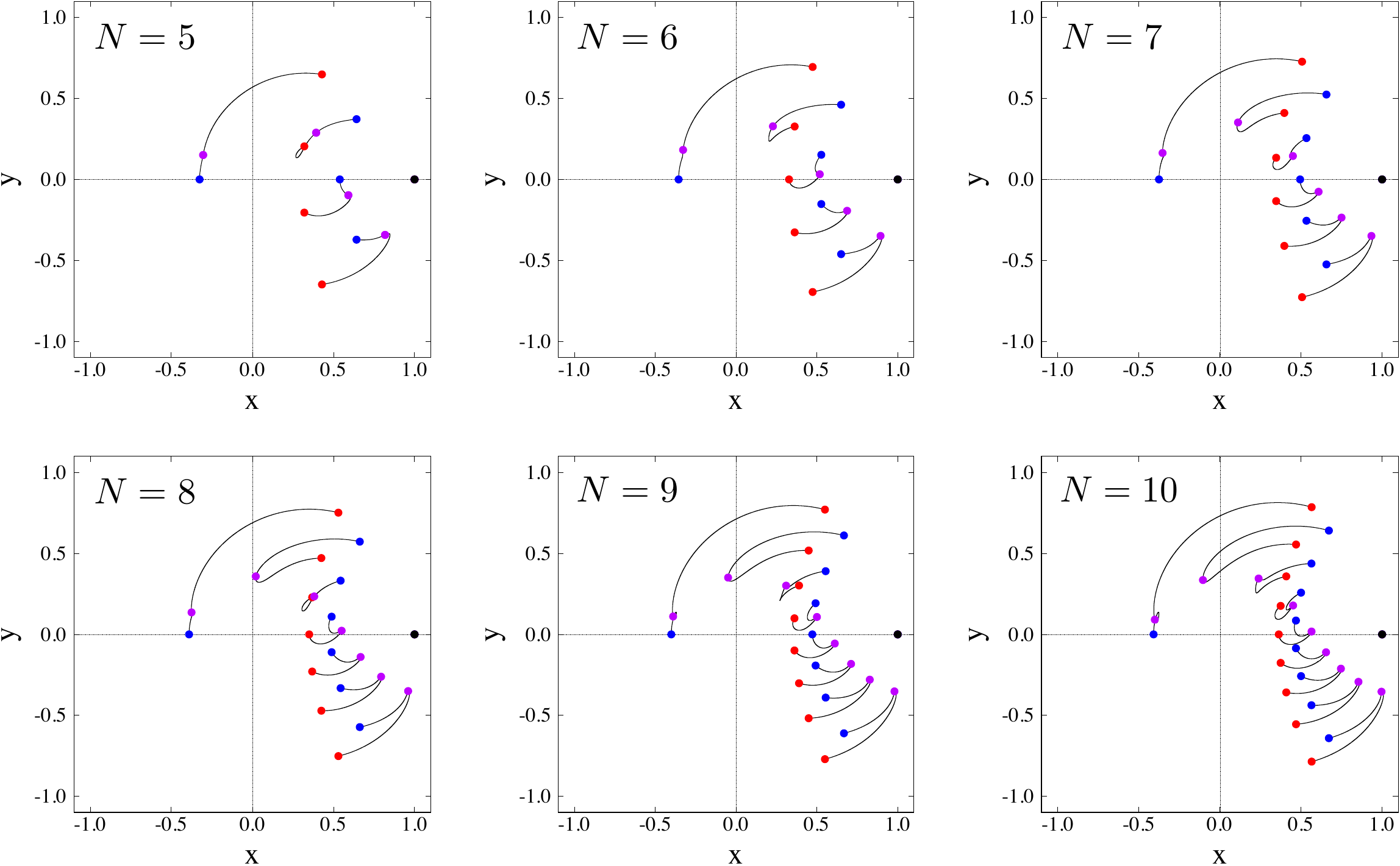}
\caption{The collapsing families for $N = 5,\cdots,10$. The black points denote the fixed $k_N$. The configurations formed by red or blue points with $k_N$ correspond to relative equilibria. The purple points correspond to configurations for $(\mathscr{H}_c, A_c)$.}
\label{fig:config-fiv-ten}
\end{center}
\end{figure}

\begin{figure}[t]
\begin{center}
\includegraphics[scale=0.6]{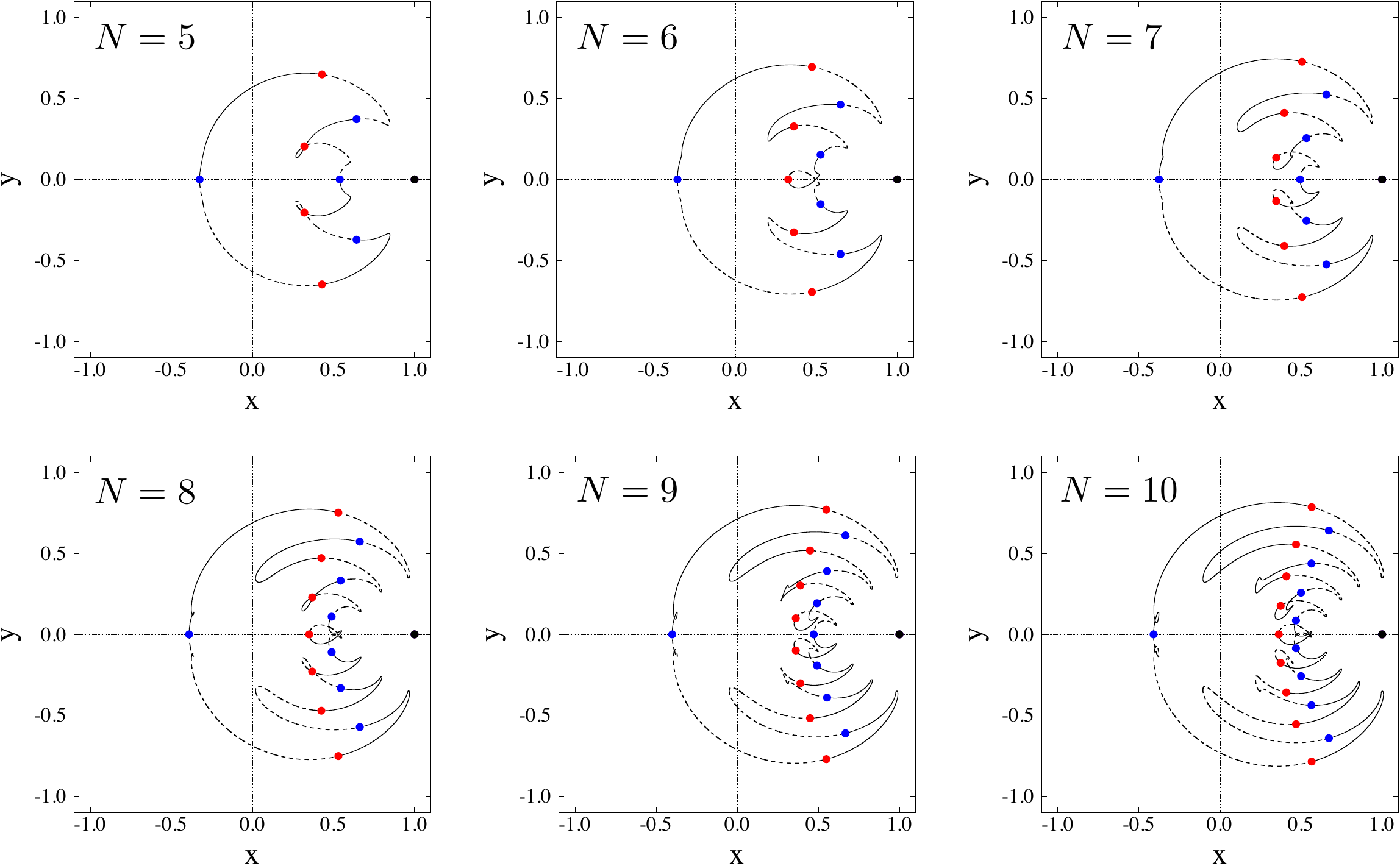}
\caption{The collapsing and expanding families for $N = 5,\cdots,10$. The solid curves denote the collapsing families, which are the same as Figure~\ref{fig:config-fiv-ten}, and the dashed ones do the expanding families.}
\label{fig:config-fiv-ten-closed}
\end{center}
\end{figure}

Next, we see the case of $N \geq 5$. Figure~\ref{fig:H-A-fiv-ten} shows the $\mathscr{H}$-$A$ curves for $N \geq 5$. In the numerical method, which we used for finding $\mathscr{H}$-$A$ curves, we just change the value of $\mathscr{H}$ forward or backward from a configuration obtained beforehand. Unlike the cases of the examples in Section~\ref{sec:exact-solution}, for each $\mathscr{H}$-$A$ curve in Figure~\ref{fig:H-A-fiv-ten}, there exists a minimal value of $\mathscr{H}$ at which we have $A \neq 0$ and the tangent of the $\mathscr{H}$-$A$ curve is parallel to the $A$-axis. We denote by $(\mathscr{H}_c, A_c)$ that point for convenience. Then, perturbing the point $(\mathscr{H}_c, A_c)$ slightly, we find another $\mathscr{H}$-$A$ curve connected to $(\mathscr{H}_c, A_c)$. This fact implies that $A$ is not determined as a single-valued function of $\mathscr{H}$. On the other hand, similarly to the examples in Section~\ref{sec:exact-solution}, we observe that configurations at the end points of $\mathscr{H}$-$A$ curves are in relative equilibria. The collapsing families corresponding to the $\mathscr{H}$-$A$ curves in Figure~\ref{fig:H-A-fiv-ten} are shown in Figure~\ref{fig:config-fiv-ten}. As we see in those figures, for each number of $N$, two configurations in relative equilibria are symmetric with respect to the $x$-axis and, as the number of $N$ increases, their point vortices seem to form a sheet except for one or two vortices. Unlike the four PV system, it is hard to prove the existence of the relative equilibria since we have to solve two or more nonlinear equations given by \eqref{v1_km}. Figure~\ref{fig:config-fiv-ten-closed} shows the collapsing families in Figure~\ref{fig:config-fiv-ten} and their reflections over the $x$-axis, that is, the expanding families. Similarly to the case of $N = 4$, the collapsing and expanding families are connected at relative equilibria and form a closed curve on the $x$-$y$ plane.

We remark that the $\mathscr{H}$-$A$ curves shown in Figure~\ref{fig:H-A-fiv-ten} are one of the examples for each $N=5, \cdots, 10$ and there is a possibility of the existence of other curves. Indeed, we have found another curve for $N = 10$, which is denoted by $C_2$ in Figure~\ref{fig:H-A-ten}$(a)$, and it gives different relative equilibria from those shown in Figure~\ref{fig:config-fiv-ten}, which is denoted by $C_1$. Figure~\ref{fig:H-A-ten}$(b)$ shows the collapsing family and the family with the expanding one that correspond to $C_2$. The figures imply that the curve formed by the collapsing and expanding families is not necessarily closed.


\begin{figure}[t]
\begin{center}
\includegraphics[scale=0.6]{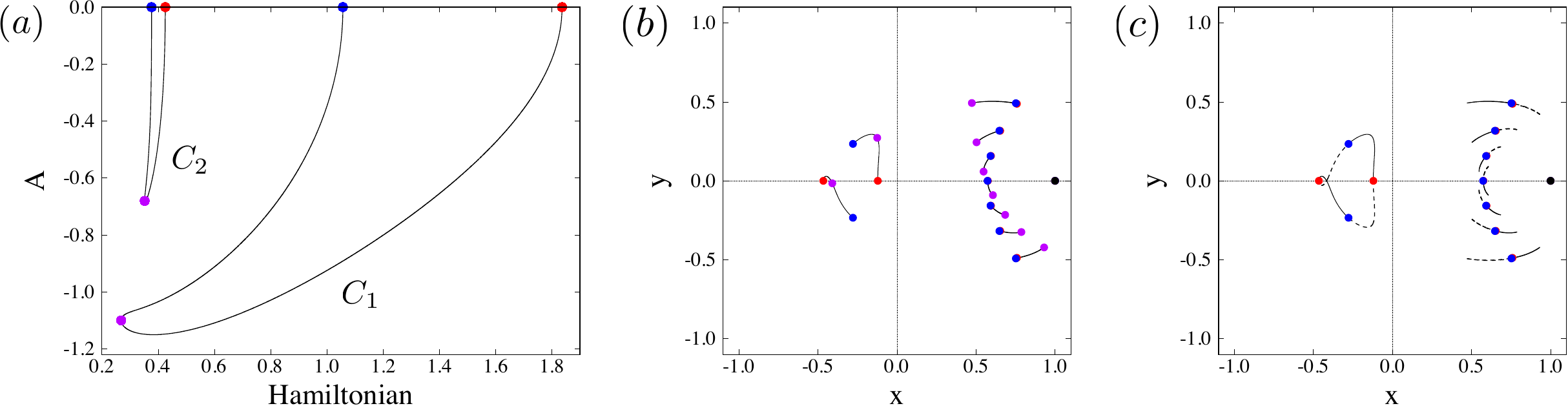}
\caption{$(a)$ The $\mathscr{H}$-$A$ curves for $N = 10$. $(b)$ The collapsing family and $(c)$ the collapsing and expanding families that correspond to $C_2$. The black points denote the fixed $k_{10}$. Note that the red and blue points with $x > 0$ are located near each other. The red or blue points correspond to relative equilibria and the purple points correspond to $(\mathscr{H}_c, A_c)$.}
\label{fig:H-A-ten}
\end{center}
\end{figure}

\subsection{Motions of point vortices with non-uniform strengths}
\label{sec:general-strength}

\begin{figure}[t]
\begin{center}
\includegraphics[scale=0.6]{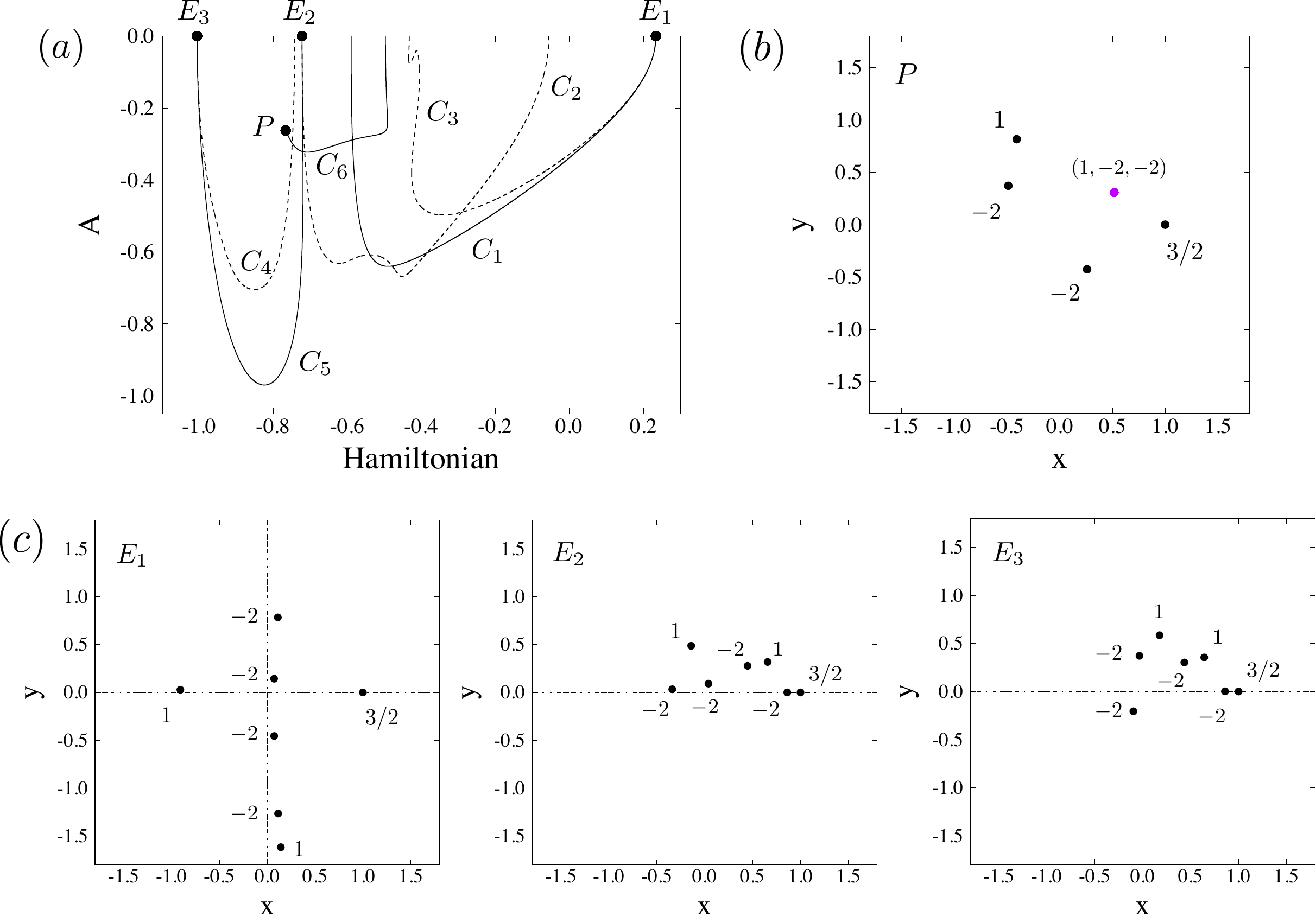}
\caption{$(a)$ The $\mathscr{H}$-$A$ curves for the seven point vortices with \eqref{g:seven-general}. $(b)$ The configuration at the point $P$ in the figure $(a)$. The purple point denotes the position at which three point vortices are located. $(c)$ The configurations in relative equilibria corresponding to $E_1$, $E_2$ and $E_3$ in the figure $(a)$.}
\label{fig:H-A-seven-general}
\end{center}
\end{figure}

We investigate self-similar collapse for non-uniform vortex strengths. For a number of $N$, examples of non-uniform vortex strengths, for which there exists a configuration leading to self-similar collapse, have been obtained by using numerical computations \cite{Kudela-1,Kudela-2,Oneil-2}. In particular, we consider the seven point vortices whose strengths are given by
\begin{equation}
\Gamma_1 = \Gamma_2 = 1, \qquad  \Gamma_3 = \Gamma_4 = \Gamma_5 = \Gamma_6 = -2, \qquad \Gamma_7 = \frac{3}{2}. \label{g:seven-general}
\end{equation}
Note that this example is given in \cite{Kudela-1}. Since point vortices that have same strengths are not distinguished, we fix the unique point vortex, that is, the seventh one with $\Gamma_7 = 3/2$ to the point $(1,0)$ for identifying equivalent configurations. Figure~\ref{fig:H-A-seven-general}$(a)$ shows the $\mathscr{H}$-$A$ curves for \eqref{g:seven-general} and there exist several $\mathscr{H}$-$A$ curves, which are not necessarily all of them. Some of curves are connected to the same points, which are denoted by $E_1$, $E_2$ and $E_3$ in the figure, on the $\mathscr{H}$-axis. At the point $E_1$, for example, the configurations in relative equilibria for $C_1$ and $C_3$ are different on the $x$-$y$ plane, see Figure~\ref{fig:config-seven-general}. However, applying a self-similar transformation or a reflection over the $x$-axis to each configurations, we find that two configurations are identical to the configuration shown in Figure~\ref{fig:H-A-seven-general}$(c)$-$E_1$. The same argument is valid for the other relative equilibria at $E_2$ and $E_3$, see Figure~\ref{fig:H-A-seven-general}$(c)$ and Figure~\ref{fig:config-seven-general}. Another remarkable feature of the example \eqref{g:seven-general} is the $\mathscr{H}$-$A$ curve labeled by $C_6$. One of the end points of $C_6$, denoted by $P$, does not satisfy $A = 0$ and the corresponding configuration, which is shown in Figure~\ref{fig:H-A-seven-general}$(b)$, is not in a relative equilibrium. We find from the figure that three point vortices are located at the same point and, in this situation, we can not find another $\mathscr{H}$-$A$ curve starting from the point $P$ by our numerical method. See Figure~\ref{fig:config-seven-general} for the entire collapsing family for $C_6$. From the above considerations, it is indicated that there are a greater number of collapsing families for non-uniform vortex strengths than uniform vortex strengths and they do not necessarily connected to relative equilibria.

\begin{figure}[t]
\begin{center}
\includegraphics[scale=0.6]{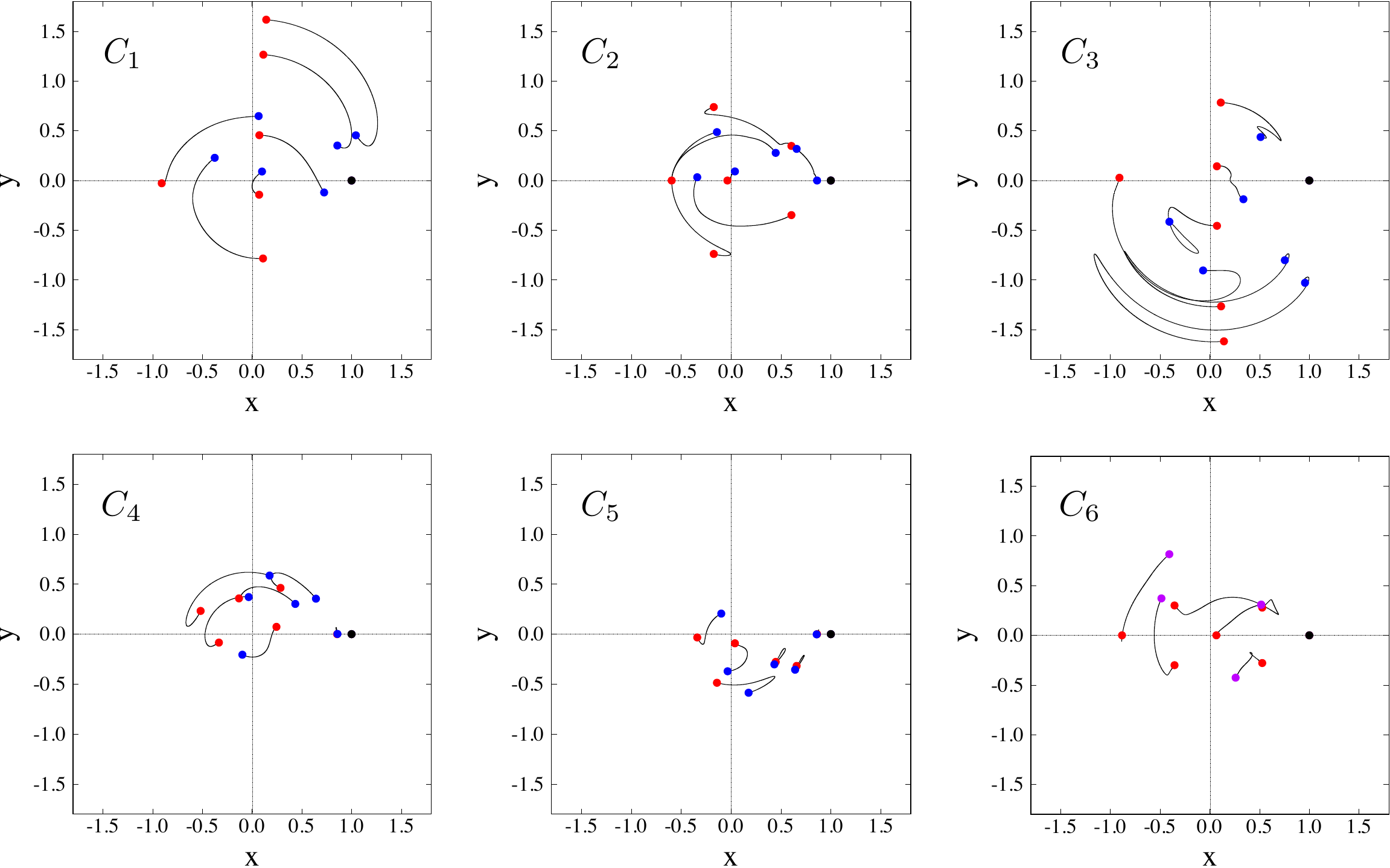}
\caption{The collapsing families corresponding to the $\mathscr{H}$-$A$ curves labeled by $C_1, \cdots C_6$ in Figure~\ref{fig:H-A-seven-general}$(a)$. The black points are $k_7$ and the configurations formed by red or blue points with $k_7$ are in relative equilibria. On each  $\mathscr{H}$-$A$ curve, the relative equilibrium for the greater value of $\mathscr{H}$ is represented by the red points and that for the smaller value of $\mathscr{H}$ corresponds to the blue one. The purple points in $C_6$ correspond to the point $P$.}
\label{fig:config-seven-general}
\end{center}
\end{figure}

\section{Concluding remarks}
\label{sec:concluding}

We have investigated self-similar collapsing solutions of the PV system. Introducing the collapsing families, which consist of configurations leading to self-similar collapse, we have shown that these families continuously depend on the Hamiltonian. Indeed, the collapse time and the Hamiltonian for the well-known exact collapsing solutions can be expressed by one-parameter functions, and the configurations at the boundary points of the parameter are in relative equilibria. For the PV problem with $N \geq 4$, we have studied the collapsing families with the help of numerical computations. Considering the case that $N - 1$ point vortices have a uniform vortex strength, we have confirmed that the collapsing families continuously depend on the Hamiltonian and their configurations asymptotically approach relative equilibria, which are symmetric with respect to the $x$-axis, as the Hamiltonian gets close to certain values. In particular, we have made the configurations in relative equilibria clear for $N = 4$. For the case of $N \geq 5$, we have shown that, unlike the examples of exact collapsing solutions, the configurations leading to self-similar collapse and the Hamiltonian are not in one-to-one correspondence in general, and there exist several collapsing families. For non-uniform vortex strengths, through an example for $N = 7$, it has been indicated that the collapsing families have a variety of structures and their configurations do not necessarily approach relative equilibria.

We discuss some future works suggested by the present study. It is worthwhile investigating whether there exist a $\mathscr{H}$-$A$ curve starting from any given relative equilibrium. Conversely, it is also of interest that any collapsing or expanding family is connected to configurations in relative equilibria or configurations for which some point vortices are located at the same point. Another interest is the number of continuous collapsing families. We have confirmed that the three PV system has a unique collapsing family, which is given by \eqref{three-init-general}, but it is unclear for the case of $N \geq 4$. Although we have numerically found some examples that have several collapsing families, further numerical or mathematical analysis is required. Moreover, as inspired by Section~\ref{sec:uniform-strength}, it is a challenging attempt to find equilibria of vortex sheets numerically by use of the point-vortex approximation, but some improvements are needed for the numerical method we used in this study.

\subsection*{Acknowledgements}
The author would like to thank Professor Yoshifumi Kimura and Professor Takashi Sakajo for productive discussions and valuable comments. This work was supported by JSPS KAKENHI Grant Number JP19J00064.


\end{document}